\documentclass[11 pt]{article}
\usepackage[utf8]{inputenc}

\usepackage[dvipsnames]{xcolor}
\usepackage{amsfonts}
\usepackage{latexsym}
\usepackage{ amssymb, bm, amsmath }
\usepackage{graphicx}
\usepackage{float}
\usepackage{multirow}
\usepackage{a4wide}
\usepackage{enumerate, todonotes}
\usepackage{url}
\usepackage{cite}
\usepackage{soul}

\usepackage[]{mdframed}
\tikzset{XOR/.style={draw,circle,append after command={
        [shorten >=\pgflinewidth, shorten <=\pgflinewidth,]
        (\tikzlastnode.north) edge (\tikzlastnode.south)
        (\tikzlastnode.east) edge (\tikzlastnode.west)
        }
    }
}

\definecolor{lightblue}{rgb}{0.68,0.85,0.9}

\usepackage{array}
\usepackage{makecell}
\setcellgapes{5pt}

\usepackage{tikz}
\usetikzlibrary{shapes, backgrounds}

\pgfdeclarelayer{edgelayer}
\pgfdeclarelayer{nodelayer}
\pgfsetlayers{background,edgelayer,nodelayer,main}

\tikzstyle{Empty}=[fill=white, draw=black, shape=circle]
\tikzstyle{Gray}=[fill={rgb,255: red,128; green,128; blue,128}, draw=black, shape=circle]
\tikzstyle{Red}=[fill=red, draw=red, shape=circle]
\tikzstyle{Square}=[fill=white, draw=black, shape=rectangle]
\tikzstyle{Blue}=[fill=cyan, draw=black, shape=circle]
\tikzstyle{BlueSquare}=[fill=cyan, draw=black, shape=rectangle]
\tikzstyle{RedSquare}=[fill=red, draw=black, shape=rectangle]
\tikzstyle{Green}=[fill=green, draw=black, shape=circle]
\tikzstyle{GreenSquare}=[fill=green, draw=black, shape=rectangle]
\tikzstyle{GraySquare}=[fill={rgb,255: red,128; green,128; blue,128}, draw=black, shape=rectangle]

\tikzstyle{Normal}=[-, draw=black, fill=none]
\tikzstyle{Arrow}=[->]
\tikzstyle{Arrow2}=[<-]

\newtheorem{theorem}{Theorem}[section]

\newtheorem{corollary}[theorem]{Corollary}

\newtheorem{problem}[theorem]{Problem}

\newcommand{\qed}{\hfill \ensuremath{\Box}}

\newenvironment{proof}{
\vspace*{-\parskip}\noindent\textit{Proof.}}{$\qed$

\medskip
}

\title{On identifying codes on oriented graphs}
\author{{\sc Soura Sena Das}$\,^{a}$ {\sc Sagnik Sen}$\,^{b}$ \\
\mbox{}\\
{\small $(a)$ Indian Statistical Institute, Kolkata, India}\\
{\small $(b)$ Indian Institute of Technology Dharwad, India}}

\date{}

\begin{document}

\maketitle

\begin{abstract}

    This article studies identifying codes in oriented graphs from a computational complexity perspective. We investigate the $\mathcal{F}$-Id Code problem, where given a simple graph $G$ and a vertex subset $C$, which induces a subgraph in the family $\mathcal{F}$, as inputs
    and ask whether it is possible to orient $G$ in such a way that $C$ becomes its oriented identifying code. 
    Focusing on the family $\mathcal{F}_d$ of $d$-regular graphs, we establish a complete dichotomy by proving that the problem is polynomial-time solvable for $d\leq1$ and NP-complete for all $d\geq2$. 
    
\medskip

\noindent \textbf{Keywords:} Oriented graphs, Identifying codes, Oriented identifying codes, Regular graphs, Complexity dichotomy.

\end{abstract}

\section{Introduction}\label{sec introduction}

First introduced in $1998$~\cite{kar-chak-lev}, an \textit{identifying code} of an undirected graph $G$ is a vertex subset $C\subseteq V(G)$ such that, for every $v\in V(G)$, the set of codewords adjacent to $v$ (including itself) is non-empty and distinct for every vertex. Identifying codes have been extensively studied in simple graphs~\cite{bousquet2015identifying,CHARON20032109,foucaud2013identifying} (also see~\cite{hudry2024honkala} for an updated survey and~\cite{LDbiblio} for an online bibliography containing over 500 related articles), motivated in particular by applications such as fault detection in multiprocessor systems. In contrast, relatively little is known about identifying codes in directed graphs, particularly oriented graphs. Cohen and Havet~\cite{cohen2018minimum} initiated the systematic study of identifying codes in oriented graphs and investigated several of their structural and algorithmic properties. In this article, we focus on the computational complexity of deciding whether a given vertex subset $C$ of an undirected graph $G$ is an oriented identifying code under some orientation of $G$, motivated by a problem posed by Cohen and Havet~\cite{cohen2018minimum}.

An \textit{oriented graph} $\overrightarrow{G}$ is a directed graph without any directed cycle of length $1$ or $2$. 
Given an oriented graph 
$\overrightarrow{G}$, its set of vertices and arcs are denoted by $V(\overrightarrow{G})$ and $A(\overrightarrow{G})$, respectively. 
Moreover, $G$ denotes 
the \textit{underlying graph} of $\overrightarrow{G}$, obtained by replacing all arcs of $\overrightarrow{G}$ by edges. On the other hand, $\overrightarrow{G}$ is an \textit{orientation} of $G$. 
Given an arc $uv$, $u$ is an \textit{in-neighbor} of $v$ and $v$ is an \textit{out-neighbor} of $u$. 
The set of all in-neighbors 
(resp., out-neighbors) of $u$ is denoted by 
$N^-(u)$ (resp., $N^+(u)$). 
Furthermore, 
the \textit{closed in-neighborhood} 
(resp., \textit{closed out-neighborhood}) of $u$ is given by 
$N^-[u] = N^-(u) \cup \{u\} \text{ (resp., } 
N^+[u] = N^+(u) \cup \{u\}\text{)}$.
Given a graph $G$ (resp., oriented graph $\overrightarrow{G}$) and any vertex subset $S$, we denote the subgraph induced by $S$ by $G[S]$ (resp., $\overrightarrow{G}[S]$).

Given an oriented graph $\overrightarrow{G}$,
a vertex subset $C \subseteq V(\overrightarrow{G})$ is an \textit{identifying code} of $\overrightarrow{G}$ if the identifier
$I(u) = N^+[u] \cap C$ is non-empty and distinct for every vertex $u \in V(\overrightarrow{G})$.
Given an undirected graph $G$, 
a vertex subset $C \subseteq V(G)$ is an \textit{oriented identifying code} of $G$ 
if there exists an orientation $\overrightarrow{G}$ of $G$ 
such that $C$ is an identifying code of $\overrightarrow{G}$.
A non-empty vertex set $C\subseteq V(G)$ is called a \textit{code} and its vertices \textit{codewords} while the vertices of $V(G)\setminus C$ are called \textit{non-codewords}.  The complexity problem of our interest is the following.

\medskip

\begin{mdframed}
\textbf{\textsc{$\mathcal{F}$-Id-Code}} (where $\mathcal{F}$ is a family of simple graphs)\\
\\
\textbf{Instance:} A graph $G$ and a set of vertices $C \subseteq V(G)$ satisfying $G[C] \in \mathcal{F}$. \\
\textbf{Question:} Is $C$ an oriented identifying code of $G$?
\end{mdframed}

Cohen and Havet~\cite{cohen2018minimum} 
showed that the 
\textsc{$\mathcal{F}$-Id Code} is NP-complete when $\mathcal{F}$, denotes the family of all simple graphs, and is polynomial-time solvable  when $\mathcal{F}$ is a family of graphs in which all but finitely many graphs $H$ satisfy
$|E(H)|=O(\log_2(\phi(|V(H)|)))$, where $\phi$ is a fixed polynomial. 
Thus, while the problem is tractable for broad classes of very sparse graphs, its complexity for dense graph families remain largely unexplored. In particular, Cohen and Havet~\cite{cohen2018minimum} asked the complexity of the $\mathcal{K}$-\textsc{Id Code} problem, where $\mathcal{K}$ is the family of all complete graphs. This naturally leads to the following general question.

\begin{problem}
Let $\mathcal{F}$ be a family of simple graphs. Then what is the computational complexity of the $\mathcal{F}$-\textsc{Id Code} problem?
\end{problem}

Motivated by this question, we study the complexity dichotomy of the $\mathcal{F}$-\textsc{Id Code} problem for $d$-regular graph families, where $\mathcal{F}$ varies over the family of $d$-regular graphs for all $d\geq0$.
Following the result of Cohen and Havet~\cite{cohen2018minimum}, we investigate how the tractability of the problem changes as the density of $G[C]$ increases. Although the density of the families of  $d$-regular graphs increase linearly with the degree, the problem becomes intractable already for $d\geq2$. That means, while the problem remains tractable when 
$G[C]$ is an independent set, or a collection of disjoint edges, it becomes intractable
when $G[C]$ is a disjoint union of cycles ($24$-cycles in our proof). 
A complete dichotomy of the problem \textsc{$\mathcal{F}_d$-Id Code}, where $\mathcal{F}_d$ denotes the family of $d$-regular graphs for all $d \geq 0$, is obtained through the following result.

\begin{theorem}\label{thm main}
Let $\mathcal{F}_d$ be the family of all $d$-regular simple graphs, for $d \geq 0$. Then the \textsc{$\mathcal{F}_d$-Id Code} problem is polynomial-time solvable if $d \leq 1$ and is NP-complete if $d \geq 2$.
\end{theorem}

As a consequence, we obtain NP-completeness for
the \textsc{$\mathcal{F}$-Id Code} problem 
where $\mathcal{F}$ is one of the following graph families: bipartite graphs, outerplanar graphs, partial $2$-tree graphs, planar graphs, and graphs having maximum degree at most $\Delta$ for every $\Delta \geq 2$.


\section{Proof of Theorem~\ref{thm main}}

The case $d=0$ is already understood. Observe that, $\mathcal{F}_0$ is the family of edgeless graphs and the complexity of \textsc{$\mathcal{F}_0$-Id Code} is already known to be polynomial-time solvable~\cite{cohen2018minimum}.

Next, we consider the case $d=1$, where the problem remains polynomial-time solvable. The proof follows the technique used by Cohen and Havet~\cite{cohen2018minimum}.

\begin{theorem}\label{thm poly}
Let $\mathcal{F}_1$ be the family of all $1$-regular graphs. Then the \textsc{$\mathcal{F}_1$-Id Code} problem is polynomial-time solvable.
\end{theorem}

\begin{proof}
    Suppose in an instance of the problem \textsc{$\mathcal{F}$-Id Code} the input is the graph $G$ and its vertex subset $C \subseteq V(G)$. Since $G[C] \in \mathcal{F}_1$, that is, $G[C]$ is an $1$-regular graph, without loss of generality we may assume that $G[C]$ is a disjoint union of $k$ $K_2$s, where $K_2$ denotes the complete graph on $2$ vertices.

    Furthermore assume that the $i^{th}$ $K_2$ is presented by the edge $e_i = a_ib_i$ for all $i \in \{1,2, \ldots, k\}$.
    Notice that, there are two orientation choices for each such edge. In particular, if the edge $e_i$ is oriented from $a_i$ to $b_i$, then the identifiers of its endpoints will be $I(a_i) = \{a_i, b_i\}$ and $I(b_i) = \{b_i\}$. Similarly, if the edge $e_i$ is oriented from $b_i$ to $a_i$, then the identifiers of its endpoints will be $I(a_i) = \{a_i\}$ and $I(b_i) = \{a_i, b_i\}$. Observe that, the identifiers on the vertices $a_i, b_i$ 
    must contain the set $\{a_i, b_i\}$ irrespective of the orientation. Moreover, $\{a_i\}$ or $\{b_i\}$ is the other identifier used (based on the orientation of $e_i$).

    Let $v \in V(G) \setminus C$ be any non-codeword. Let $I^*(v) = N(v) \cap C$, that is the set of neighbors of $v$ in $C$. Notice that, given any subset $S \subseteq I^*(v)$ it is possible to orient $G$ in a way 
    that the out-neighbors of $v$ in $C$ are exactly the vertices from $S$, and thus, the identifier of $v$ becomes $I(v) = S$. Moreover, it is possible to choose the orientations of the edges between $v$ and $C$ independently from the choice of orientations of the edges between any other vertex $u \in V(G) \setminus C$ and $C$. 

    Thus, if it is possible to greedily assign a subset $S$ of $I^*(v)$ to each vertex $v$ in such a way that $S$ has not been assigned to any of the vertices of $V(G) \setminus C$ earlier, and also it doesn't feature as an identifier of a vertex from $C$, then we are done. Note that, while assigning the set $S$ to $v$, if we mandate $S \neq \{a_i, b_i\}, \{a_i\}, \{b_i\}$ for all $i \in \{1, 2, \ldots, k\}$, then we know that $S$ does not feature as an identifier of a vertex from $C$. In particular, we know that, $\{a_i, b_i\}$  surely features as an identifier of a vertex in $C$. 

    Let 
    $$IC = \{\{a_i, b_i\}, \{a_i\}, \{b_i\} : i \in \{1, 2, \ldots, k\}\}
    \text{ and } P(v) = \{S : S \subseteq I^*(v) \text{ and } S \neq \emptyset\} \setminus IC$$
    for all $v \in V(G) \setminus C$. Notice that if $|P(v)| > |V(G)|$ for some $v$, it is always possible to choose a suitable $S \in P(v)$ for $v$ in a greedy algorithm. Thus, we need not think much about such vertices. 
    Irrespective of the size of $P(v)$, it is possible to calculate its cardinality in $O(|V(G)|)$ time. The way we calculate it is the following.
    First note that the set $I^*(v)$ can be built in $O(|V(G)|)$ time. 
    If $|I^*(v)| = t$, then the cardinality of the power set of $I^*(v)$ is $2^t$. The set $P(v)$ is the power set of $I^*(v)$, excluding the empty set and the elements of $IC$.
    Note that, one can build the set $IC$ in $O(|V(G)|)$ time, check if a set of $IC$ is a subset of $I^*(v)$ also in $O(|V(G)|)$ time. 
    Let $s$ be the number of sets in $IC$ that are subsets of $I^*(v)$. Therefore, we will have $|P(v)| = 2^t - s -1$, calculated in $O(|V(G)|)$ time.

    Next consider the graph $G^*$ obtained from $G$ by deleting the vertices 
    $v \in V(G) \setminus C$ satisfying 
    $|P(v)| > |V(G)|$. Notice that, for the vertices $u \in V(G^*) \setminus C$, it is possible to build $P(u)$ in $O(|V(G)|)$ time. Since there can be at most $|V(G)|$ many such vertices, it is possible to build an auxiliary graph $H$ in $O(|V(G)|)$ time as described in the following.

The graph $H$ is a bipartite graph with the vertices of $V(G^*) \setminus C$ forming one of the partite sets, say $A$. The other partite set, say $B$, is formed by the union of the sets from $P(u)$, 
where $u$ varies in $V(G^*) \setminus C$. 
Furthermore, include $k$ more vertices in the partite set $B$, namely, $z_1, z_2, \cdots, z_k$. If a vertex $S$ of $B$ satisfies $S \in P(u)$ for some $u \in A$, then put an edge between $u$ and $S$. Also, if $u \in A$ is a neighbor of $a_i$ or $b_i$ in $G$ for some $i \in \{1,2,\ldots,k\}$, then add an edge between $u$ and $z_i$. This way our 
auxiliary graph $H$ is constructed. Notice that, $|V(H)| = O(|V(G)|^2)$.

Next try to find a matching in $H$ that saturates $A$. If we cannot find one, that means $C$ is not an oriented identifying code of $G^*$, and hence of $G$. If we can find such a matching, then we will show how to find a suitable orientation $\overrightarrow{G^*}$ of $G^*$ for which $C$ is an identifying code of $\overrightarrow{G^*}$. Later we will extend this orientation  to obtain an 
orientation $\overrightarrow{G}$ of $G$ such that $C$ is an identifying code of 
$\overrightarrow{G}$. 

Therefore, the next main objective is 
to understand how to orient $G^*$ if a suitable matching of $H$ is found. Suppose a matching of $H$ which saturates $A$ is obtained. If any vertex $u \in A$ is matched with a vertex $S \in B \setminus \{z_1, z_2, \cdots, z_k\}$, then orient the edges connecting $u$ to the vertices of $C$ in such a way that exactly the vertices of 
$S$ becomes out-neighbors of $u$. If any vertex $u \in A$ is matched with $z_i$, for some $i \in \{1,2, \cdots, k\}$, 
then there can be two scenarios. 

Firstly, if $u$ is adjacent to $a_i$, then orient the edge $e_i$ from $a_i$ to $b_i$. Also, orient  the edges between $u$ and the vertices of $C$ in such a way that only $a_i$ is the out-neighbor of $u$ in $C$.  

Secondly, if $u$ is not adjacent to $a_i$, and hence adjacent to $b_i$, then orient the edge $e_i$ from $b_i$ to $a_i$. Also,
orient the edges between $u$ and the vertices of $C$ in such a way that only $b_i$ is the out-neighbor of $u$ in $C$.  

All other edges of $G^*$ can be oriented arbitrarily. Let this so-obtained orientation be $\overrightarrow{G^*}$. 
Now add the vertices of $V(G) \setminus V(G^*)$. Also, since any $v \in V(G) \setminus V(G^*)$ satisfies $|P(v)| > |V(G)|$, it is possible to assign a subset $S$ of $P(v)$ to $v$  which is different from any identifier of a vertex in $G^*$, or any set earlier assign to a vertex from 
$V(G) \setminus V(G^*)$ using a greedy algorithm. After assigning such distinct sets, orient the edges between $V(G) \setminus V(G^*)$ and $C$ in such a way that for any vertex $v$, its out-neighbors in $C$ are exactly the vertices from the set $S$ assigned to it. The remaining edges can be oriented arbitrarily. Let the so-obtained orientation of $G$ be $\overrightarrow{G}$. 

The proof is completed by noticing $C$ is an identifying code of $\overrightarrow{G}$. 
\end{proof}

A Boolean formula is in \textit{$(3,4)$-SAT} if it is in conjunctive normal form, every clause contains exactly three literals, and each variable appears in at most four clauses. 
The corresponding decision problem asks whether there exists a truth assignment that satisfies all clauses of the formula. Tovey~\cite{tovey1984simplified} proved that \textsc{$(3,4)$-SAT} is NP-complete.

For every fixed $d \geq 0$, the problem $\mathcal{F}_d$-Id Code is in NP. To see this, observe that an orientation of $G$ serves as a polynomial-size certificate. Given an orientation, one can compute $N^+[v] \cap C$ for each $v \in V(G)$ and verify in polynomial time that the resulting identifiers are nonempty and pairwise distinct.
We are now ready to prove the NP-completeness of the \textsc{$\mathcal{F}_2$-Id Code} problem.

\begin{theorem}\label{thm 2-regular NP-c}
Let $\mathcal{F}_2$ be the family of all $2$-regular graphs. Then the \textsc{$\mathcal{F}_2$-Id Code} problem is NP-complete.
\end{theorem}

\begin{proof}
    Given a $(3,4)$-SAT formula $\phi$, we will now describe the construction of a graph $G_{\phi}$, and a vertex subset $C \subseteq V(G_{\phi})$, where $G_{\phi}[C] \in \mathcal{F}_2$. Furthermore, we will show that 
    $\phi$ is satisfiable if and only if 
    $C$ is an oriented identifying code of $G_{\phi}$. 

Suppose $\phi$ has $r$ variables, namely, 
$x_1, x_2, \cdots, x_r$ and 
$\ell$ clauses, namely, 
$C_1, C_2, \cdots, C_{\ell}$. 

\medskip

\noindent \textbf{Construction of the variable gadgets:} For each variable $x_i$, construct a $24$-cycle of the form 
$$A_i = a_{i,1} a_{i,2}  \cdots a_{i,24}a_{i,1},$$ 
where $i \in \{1,2, \ldots, r\}$. The vertices of these cycles are part of the set $C$. 
Next, for each edge $a_{i,j}a_{i,j+1}$ (the second coordinates of the indices are considered modulo $24$)
of the cycle $A_i$ add one vertex $b_{i,j}$  adjacent to the vertices 
$a_{i,j}$ and $a_{i,j+1}$. 
After that, for even values of $j$, add 
another new vertex $b'_{i,j}$ to the vertices $a_{i,j}$ and $a_{i,j+1}$. The vertices of the form $b_{i,j}$ and $b'_{i,j}$ are part of $V(G) \setminus C$.

This completes the construction of the variable gadget. 

\medskip

\noindent \textbf{Construction of the clause gadgets:}
The vertices of the cycle $A_{i}$ are partitioned into $8$ sets called blocks, denoted by, $B_{i,1}, B_{i,2}, \cdots, B_{i, 8}$, where $B_{i,j} = \{a_{i, 3j-2}, a_{i,3j-1}, a_{i,3j}\}$, where $j \in \{1, 2, \ldots, 8\}$. Moreover, the vertex $a_{i, 3j-1}$ is called the center of the block $B_{i,j}$. 
A block $B_{i,j}$ is odd (resp., even) if $j$ is odd (resp., even). Let us keep in mind that the odd blocks will correspond to the positive occurrence of a variable, and the even blocks will correspond to the negative occurrence of a variable.

For a clause $C_q$, add a vertex $z_q$ in the set $V(G) \setminus C$, where $q \in \{1, 2, \ldots, \ell\}$. 
Suppose the three variables that occur  
in $C_q$ are $x_{i_1}, x_{i_2}, x_{i_3}$. 
Then we will make $z_q$ adjacent to the vertices of exactly one block from each of $A_{i_1}, A_{i_2}, A_{i_3}$. The choice of the blocks will depend on whether $x_{i_1}$ (resp., $x_{i_2}, x_{i_3}$) appears as a positive or a negative in $C_q$. 
If $x_{i_1}$ (resp., $x_{i_2}, x_{i_3}$) appears as a positive in $C_q$, then 
make $z_q$ adjacent to the vertices of an odd block of $A_{i_1}$  
(resp., $A_{i_2}, A_{i_3}$). 
If $x_{i_1}$ (resp., $x_{i_2}, x_{i_3}$) appears as a negative in $C_q$, then 
make $z_q$ adjacent to the vertices of an even block of $A_{i_1}$  
(resp., $A_{i_2}, A_{i_3}$).

During this entire procedure, make sure to use separate blocks for different $z_q$s, where $q \in \{1, 2, \ldots, \ell\}$. This is possible as each variable can occur in maximum four clauses, and each $A_i$ has exactly four odd blocks, and exactly four even blocks.

Recall that, a false-twin of a vertex $u$ is another vertex $v$ such that $N(u) = N(v)$, and $u$ and $v$ are non-adjacent. 
Finally, add exactly $493$ false-twins of  $z_q$ in the set $V(G) \setminus C$, where $q \in \{1, 2, \ldots, \ell\}$. Name these 
false twins as $z_{q, 1}, z_{q, 2}, \cdots, z_{q, 493}$. 
Furthermore, let $Z_q = \{z_q, z_{q, 1}, z_{q, 2}, \cdots, z_{q, 493}\}$. The vertex set $Z_q$ is the gadget for the 
clause $C_q$. 

This completes the construction of the clause gadget.

\medskip

\noindent\textit{Claim 1.}\quad  
Suppose  $\overrightarrow{G_{\phi}}$ is an 
orientation of $G_{\phi}$ for which  $C$ 
is an identifying code. Let $\overrightarrow{G_{\phi}}[A_i]$ denote the oriented cycle induced by the vertices of $A_i$, for all $i \in \{1,2, \ldots, r\}$. 
Then the following are true.

\begin{enumerate}[(a)]
    \item The vertices of $\overrightarrow{G_{\phi}}[A_i]$ are either source or sink. 

    \item The oriented cycle $\overrightarrow{G_{\phi}}[A_i]$ has exactly $12$ sources, and $12$ sinks.

    \item The oriented cycle $\overrightarrow{G_{\phi}}[A_i]$ can have exactly one out of the 
    two following orientations: 
    (i) all vertices with odd indices in its second coordinate are sources, and all vertices with even indices in its second coordinate are sinks; 
    (ii) all vertices with odd indices in its second coordinate are sinks, and all vertices with even indices in its second coordinate are sources.
\end{enumerate}

\noindent \textit{Proof of Claim.}
Suppose $\overrightarrow{G_{\phi}}$ is an 
orientation of $G_{\phi}$ for which  $C$ 
is an identifying code.

An arbitrary vertex $a_{i,j}$ of $A_i$ 
have exactly four options of identifiers, namely, 
$\{a_{i, j}\}$, 
$\{a_{i,j-1}, a_{i,j}\}$,
$\{a_{i,j}, a_{i, j+1}\}$,
$\{a_{i,j-1}, a_{i,j}, a_{i, j+1}\}$. 
Let us take the union of these four identifier options for
$a_{i,j}$ while varying $j$ from $1$ to $24$. Notice that, the adjacent vertices will have a couple of common options, and overall, the union of these options are 
$72$ (that is, $24$ singleton sets, $24$ doubleton sets, and $24$ three element sets). Let $S_i$ denote set of all such $72$ 
sets. Moreover, let $S_{i,1}$ denote the set of all singleton sets in $S_i$, $S_{i,2}$ denote the set of all doubleton sets in $S_i$, and $S_{i,3}$ denote the set of all three element sets in $S_i$.  

Notice that there are $36$ vertices of the form $b_{i,j}$ and $b'_{i,j}$. These vertices must have identifiers from $S_i$ 
since they are adjacent to the end points of the edges of $A_i$. 
Moreover, since these vertices are adjacent to exactly two vertices of  $C$ (both from $A_i$), their identifiers must be either a singleton or a doubleton set from $S$. 

That means, $36$ among the $48$ sets in $S_{i,1} \cup S_{i,2}$ must be used as identifiers of the vertices of the form $b_{i,j}$ and $b'_{i,j}$. That means, we can use a maximum of $(48 -36) = 12$ from $S_{i,1} \cup S_{i,2}$ sets 
as identifiers for the $24$ vertices of $A_i$. 
That also implies that, we need to use at least $12$ sets from $S_{i,3}$ as identifiers for the $24$ vertices of $A_i$. Notice that, if a set from $S_{i,3}$ is used as an identifier of a vertex $a_{i,j}$, then $a_{i,j}$ must be a source in $\overrightarrow{G_{\phi}}[A_i]$. 
That means, $\overrightarrow{G_{\phi}}[A_i]$ has at least $12$ sources. 

Observe that, the number of sources and sinks in $\overrightarrow{G_{\phi}}[A_i]$ must be the same. The reason is as follows. Let $p$ (resp., $p'$) 
be the number of sources (resp., sinks) in $\overrightarrow{G_{\phi}}[A_i]$. 
That means, $\overrightarrow{G_{\phi}}[A_i]$
has exactly $p$ vertices having out-degree 
equal to $2$ and in-degree equal to $0$;
exactly $p'$ vertices having out-degree 
equal to $0$ and in-degree equal to $2$;
and exactly $(24 - p - p')$ vertices having out-degree 
equal to $1$ and in-degree equal to $1$. 
Since the sum of out-degrees and in-degrees in any directed graph is equal due to the Handshaking Lemma for directed graphs, we have 
$$2p+(24-p-p') = 2p' + (24-p-p') \implies p = p'.$$

That means, the maximum number of sources 
we can have is $12$. Therefore, in our scenario, 
$\overrightarrow{G_{\phi}}[A_i]$ must have exactly $12$ sources and $12$ sinks.

This proves parts (a) and (b) of the claim. 
The part (c) of the claim follows from the fact that every vertex of the oriented $24$-cycle is a source or a sink, and hence sources and sinks must alternate around the cycle.
\hfill $\circ$

\medskip

\noindent \textbf{The equivalence (forward direction):}  Suppose $\phi$ is satisfiable. That means, the variable have an truth or false assignment so that every clause contains at least one true literal. Let us find an orientation of $G_{\phi}$ where $C$ is an identifying code. 

If $x_i$ is assigned truth, then orient $A_i$ is such a way that all the center of its odd blocks are sinks, and the center of all its even blocks are sources.
Similarly, if $x_i$ is assigned false, then orient $A_i$ is such a way that all the center of its odd blocks are sources, and the center of all its even blocks are sinks. These are 
 valid orientations since all the centers of the odd blocks (resp., even blocks) have indices of the same parity, and due to Claim~1(c). 

Next recall the definitions of $S_i, S_{i,1}, S_{i,2}, S_{i,3}$ from the proof of Claim~1. Orient the edges incident to the vertices corresponding to $b_{i,j}$s as sources. That means, the identifiers of $b_{i,j}$s will exhaust the sets belonging to $S_{i,2}$. 
Notice that, $b'_{i,j}$s have two neighbors from $A_i$, one with an odd index another with an even index. 

If the vertices having odd (resp., even) indices in $A_i$ are sinks, then use the singleton sets from $S_{i,1}$ 
containing vertices with even (resp., odd) indices as identifiers of the vertices of the type $b'_{i,j}$, and orient the edges incident to them accordingly. This is possible since $b'_{i,j}$s have distinct sets of neighbors. 
This means, all identifiers from the set 
$S_{i,1} \cup S_{i,2}$ has been exhausted. 

Since different clause gadgets use distinct blocks, identifiers assigned to
vertices of different clause gadgets are automatically distinct.
Now it remains to orient the edges 
incident to the vertices of $Z_q$, for all $q \in \{1,2,\ldots, \ell\}$. Let us fix some $q \in \{1,2,\ldots, \ell\}$.
Suppose that the vertices of $Z_q$ is adjacent to vertices of three different blocks, namely, $B_1$, $B_2$, $B_3$ belonging to three different $24$-cycles 
corresponding to the variable gadgets. 
Thus, potentially, there are $2^9-1 = 511$ different identifiers that can be assigned to them. However, three singleton 
and two doubleton subsets of the vertices of $B_1$ are already used as identifiers
(in the $24$-cycle containing $B_1$). 
 That means, $5$ potential identifiers for the vertices of $Z_q$, which are subsets of $B_1$ have already been used. Similarly, $5$ potential identifiers for the vertices of $Z_q$, which are subsets of $B_2$ (resp., $B_3$) have already been used. That means, the potential identifiers for the vertices of $Z_q$ is reduced to $(511 - 15) = 496$. 

 Furthermore, note that, if the center of $B_1$ (resp., $B_2$, $B_3$) is a source, then the unique three element subset of the vertices of $B_1$ (resp., $B_2$, $B_3$) has also been used as an identifier. Since the clause $C_q$ has at least one true literal, not all centers of $B_1, B_2, B_3$ can be sources. That means, we will be left with at least $494$  potential identifiers of the vertices of $Z_q$. Thus, assign these set of $494$ identifiers distinctly to the $494$ vertices of $Z_q$, and orient the edges incident to them accordingly. 
 
 The so-obtained orientation of $G_{\phi}$ is our desired $\overrightarrow{G_{\phi}}$. The way the orientation 
 $\overrightarrow{G_{\phi}}$ is obtained, 
 one can observe that $C$ is an identifying code of $\overrightarrow{G_{\phi}}$.

\medskip

\noindent \textbf{The equivalence (backward direction):} Suppose 
there exists an orientation $\overrightarrow{G_{\phi}}$ of $G_{\phi}$ such that $C$ is an identifying code of 
$\overrightarrow{G_{\phi}}$. 
Our goal is to provide a truth assignment to the variables $x_1, x_2,\ldots, x_r$ which satisfies the formula $\phi$.

According to Claim~1(c), the $24$-cycles of the form $A_i$ can have one of the two following possible orientations: 
(i) all vertices with odd indices in its second coordinate are sources, and all vertices with even indices in its second coordinate are sinks; 
(ii) all vertices with odd indices in its second coordinate are sinks, and all vertices with even indices in its second coordinate are sources.
In case, if the orientation of $\overrightarrow{G_{\phi}}[A_i]$ is of type (i), then assign true to the variable $x_i$. If the orientation of $\overrightarrow{G_{\phi}}[A_i]$ is of type (ii), then assign false to the variable $x_i$. 
In other words, the orientation $\overrightarrow{G_\phi}[A_i]$ corresponds to the assignment of truth (resp., false) to the variable $x_i$ if the centers of the odd (resp., even) blocks are sinks. 
We will show that this assignment satisfies $\phi$. To prove that it is enough to show that each clause contains at least one true literal.

Recall the definitions of $S_i, S_{i,1}, S_{i,2}, S_{i,3}$ from the proof of Claim~1.
Observe that, irrespective of the orientation of $\overrightarrow{G_{\phi}}$, the sets from $S_{i,1} \cup S_{i,2}$ are already used as identifiers of the vertices of $A_i$ and vertices of the form $b_{ij}$ and $b'_{ij}$.
This follows from the fact that due to Claim 1, there are exactly $12$ sinks in $A_i$, which use $12$ distinct
singleton identifiers. Moreover, the $36$ vertices of the form $b_{i,j}$ and
$b'_{i,j}$ use $36$ further distinct identifiers from $S_{i,1}\cup S_{i,2}$.
Since $|S_{i,1}\cup S_{i,2}|=24+24=48$,
all identifiers in $S_{i,1}\cup S_{i,2}$ have already been used.
Let us now count the options of identifiers of the vertices of $Z_q$ for some $q \in \{1,2, \ldots, \ell\}$.
The vertices of $Z_q$ are false-twins, and thus are independent vertices having the exact same set of neighbors. They have $9$ neighbors, and thus a total of $2^9-1 = 511$ options for assigning distinct identifiers.

Since the identifiers belonging to sets of the type 
$S_{i,1} \cup S_{i,2}$ are already used, and since the vertices of $Z_q$ are adjacent to vertices of exactly three distinct blocks, $(5 \times 3) = 15$ more options are unavailable for the vertices of $Z_q$ that we had counted earlier. 
That leaves us with $(511-15) = 496$ options. 

Let $B$ be a block whose vertices are adjacent to the vertices of $Z_q$. 
If the center of $B$ is a source, then 
a three element set, among the $496$ options for the identifiers of the vertices of $Z_q$ counted above,   is also eliminated. Since we know that all vertices of $Z_q$ has distinct identifiers, that is, collectively they have $494$ distinct identifiers, 
not all three centers of the blocks whose vertices are adjacent to the vertices of $Z_q$ can be sources. This corresponds to  every clause containing a true literal.

\medskip

\noindent \textbf{Counting the number of vertices of $G_{\phi}$:}
Corresponding to each variable $x_i$, there is a $24$-cycle $A_i$, $24$ vertices of the form $b_{i,j}$, $12$ vertices of the form $b'_{i,j}$, for all $i \in \{1,2,\ldots,r\}$. That means, totally, there are $(24+24+12)r = 60r$ vertices
used to construct the variable gadgets.

Corresponding to each clause $C_q$, there are $494$ vertices from $Z_q$
for all $q \in \{1,2,\ldots,\ell\}$. That means, totally, there are $494\ell$ vertices  that are used to construct the clause gadgets. 

Thus, the graph $G_{\phi}$ contains $60r+494\ell$ vertices, which is linear in terms of the input formula size. 

\medskip

\noindent \textbf{Concluding the proof:}
For any $(3,4)$-SAT formula $\phi$ we constructed a graph $G_{\phi}$ (along with designating vertex subset $C$) having 
number of vertices linear in the size of $\phi$ satisfying the following: 
the formula $\phi$ is satisfiable if and only if $C$ is an oriented identifying code of $G_{\phi}$. This completes the proof as \textsc{$(3,4)$-SAT} NP-complete.
\end{proof}

The proof of the next result uses the 
above theorem. To be precise, we will use the following  corollary of it.

\begin{corollary}\label{cor 2-regular bipartite}
  Let $\mathcal{B}_2$ be the family of all $2$-regular bipartite graphs having even number of vertices. Then the \textsc{$\mathcal{B}_2$-Id Code} problem  is NP-complete.
\end{corollary}

\begin{proof}
    Follows directly from Theorem~\ref{thm 2-regular NP-c} since $G_{\phi}[C]$ from the proof is a disjoint union of $24$-cycles, which are in particular $2$-regular bipartite graphs.  
\end{proof}

\begin{theorem}\label{thm d-regular NP-c}
Let $\mathcal{F}_d$ be the family of all $d$-regular graphs, for $d \geq 3$. Then the \textsc{$\mathcal{F}_d$-Id Code} problem  is NP-complete.
\end{theorem}

\begin{proof}
    For this proof, we are going to use 
    Corollary~\ref{cor 2-regular bipartite}. 
    In particular, given any  graph $G$ and a vertex subset $C$ such that $G[C]$ is a $2$-regular bipartite graph, we will construct a graph $G^*$ with a vertex subset $C^*$ such that $G^*[C^*]$ is a $d$-regular graph. Moreover, we will show that 
    $C^*$ is an oriented identifying code of $G^*$ if and only if $C$ is an oriented identifying code of $G$.

First we are going to describe the construction of some structures which we call agents. These agents will be useful to construct $G^*$ from $G$. In particular, $G^*$ will be a supergraph of $G$ and $C^*$ will be a superset of $C$. The construction of the agents and how they connect to the vertices of $G$ are crucial aspects of the construction of $G^*$ from $G$.

\medskip

\noindent \textbf{Construction of agents:}  Start with a $K_{d+1} - e$, 
which denotes the graph obtained by deleting one edge from the complete graph on $(d+1)$ vertices. Suppose $y_1, y_2$ are the unique non-adjacent pairs of vertices of the $K_{d+1} - e$. 
The vertices $y_1, y_2$ of $K_{d+1} - e$ are called the pivots, and the rest of the vertices are called the basic vertices. 
Let $B$ denote the set of basic vertices. We first add $2^{d-1}-1$ vertices, each adjacent to every vertex in $B$. We then add two sets of $2^{d-1}$ new vertices, where the vertices in the first set are adjacent to $y_1$ and every vertex of $B$, while those in the second set are adjacent to $y_2$ and every vertex of $B$. We denote the resulting construction by $Y$ and refer to it as an \textit{agent}.

\medskip

\noindent \textbf{Construction of $G^*$:}  We know that $C$ has even number of vertices. 
Partition the vertices of $C$ into $\frac{|C|}{2}$ disjoint pairs. 
Let $\{u_1, u_2\}$ be such a pair. 
For each pair $\{u_1, u_2\}$ 
take $(d-2)$ agents and connect them to $u_1$ and $u_2$. 
The way we connect an agent $Y$ to a pair of vertices 
$\{u_1, u_2\}$ is through adding the edges $u_1y_1$ and $u_2y_2$. This completes the construction of the graph $G^*$.

\medskip

\noindent \textbf{Declaration of $C^*$:} 
The set $C^*$ is a union of the vertices of $C$ and the pivots and the base vertices of the agents used to construct $G^*$. The rest of the vertices of the agents belong to $V(G^*) \setminus C^*$.

\medskip

\noindent \textbf{The equivalence (forward direction):} Now we are going to prove that if
$C^*$ is an oriented identifying code of $G^*$, then $C$ is an oriented identifying code of $G$. 

First suppose $C^*$ is an oriented identifying code of $G^*$. 
That means, there exists an orientation $\overrightarrow{G^*}$ such that $C^*$ is an identifying code. Let $\overrightarrow{G}$ denote the graph induced by the vertices of $G$ from the oriented graph $\overrightarrow{G^*}$. 
Notice that the vertices of $G^*$ are a union of the vertices of $G$ and the vertices of the agents. 
The vertices of 
$V(G) \setminus C$ 
are not adjacent to any vertex of the agents. 
That means, the identifiers of the vertices of $V(G) \setminus C$ are the same in $\overrightarrow{G^*}$
and 
$\overrightarrow{G}$. 

Next, if we can show that the identifiers of the vertices of $C$ are the same in 
$\overrightarrow{G^*}$
and 
$\overrightarrow{G}$, then it will imply that $C$ is an identifying code of 
$\overrightarrow{G}$.

Let $Y$ be an agent with pivot vertices $y_1, y_2$ and a set $B$ of base vertices. Since $|B|=d-1$, the sets $\{y_1\}\cup B$ and $\{y_2\}\cup B$ each have $2^d$ subsets. Observe that the $2^{d-1}$ subsets of $B$ are common to both collections. Thus, the total number of distinct subsets of $\{y_1\}\cup B$ or $\{y_2\}\cup B$ is
$
2^d+2^d-2^{d-1}=2^{d+1}-2^{d-1}.
$
Excluding the empty set, there are $2^{d+1}-2^{d-1}-1$ possible non-empty identifiers. This is exactly the number of such vertices in $Y$, since
$
(2^{d-1}-1)+2^{d-1}+2^{d-1}
=2^{d+1}-2^{d-1}-1.
$
Therefore, these vertices must use all non-empty subsets of $\{y_1\}\cup B$ and $\{y_2\}\cup B$ as identifiers.


Since the vertices of $B$ must be assigned distinct identifiers, each of them must have both $y_1, y_2$ in their identifier sets. Moreover, that will force $y_1$ (resp., $y_2$) to have their only neighbor from the set $C$ to be in their identifier set. That means, all edges between the vertices of $C$ and the agents are oriented from the agents to $C$. Thus, none of the identifiers of $C$ contains any vertex from outside of $G$. 
Thus, 
the identifiers of the vertices of $C$ are the same in 
$\overrightarrow{G^*}$
and 
$\overrightarrow{G}$. Hence $C$ is an oriented identifying code of $G$.

\medskip

\noindent \textbf{The equivalence (backward direction):} Now we are going to prove that if
$C$ is an oriented identifying code of $G$, then $C^*$ is an oriented identifying code of $G^*$.

First suppose $C$ is an oriented identifying code of $G$. 
That means, there exists an orientation $\overrightarrow{G}$ such that $C$ is an identifying code. We will try to find an orientation $\overrightarrow{G^*}$ of $G^*$ such that $C^*$ is an identifying code of $\overrightarrow{G^*}$. 

To obtain $\overrightarrow{G^*}$, we will 
first retain the orientations of the oriented edges from $\overrightarrow{G}$. 
After that orient all the edges from the vertices of the agents to the vertices of $G$. For any agent, orient the edges from the vertices  of the base to their pivots. The orientations of the edges having base vertices as end points can be oriented arbitrarily.

Suppose $Y$ is an arbitrary agent with pivots $y_1,y_2$ and base vertices $B$. The $2^{d-1}-1$ vertices adjacent only to $B$ are assigned distinct identifiers from the non-empty subsets of $B$. The remaining $2^{d-1}$ vertices adjacent to $y_1$ and $B$ are assigned the identifiers containing $y_1$, while the $2^{d-1}$ vertices adjacent to $y_2$ and $B$ are assigned the identifiers containing $y_2$. We orient the corresponding edges so that these assigned sets become the identifiers of the respective vertices.


Observe that $C^*$ is indeed an identifying code of the so-obtained $\overrightarrow{G^*}$. 

\medskip

\noindent\textbf{Counting the number of vertices of $G^*$:}
Suppose $G$ had $n$ vertices, and $C$ had $k$ vertices. Notice that there are $2$ pivots, $(d-1)$ base vertices, and $2(2^d-1)$ other vertices in every agent. Also, there are a total of $\frac{(d-2)k}{2}$ agents. Hence the total number of vertices in $G^*$ is 
$$n+\frac{(d-2)k}{2}(2+(d-1)+2(2^d-1))  = O(n),$$
since $k \leq n$ and $d$ is a constant.

\medskip

\noindent \textbf{Concluding the proof:}
 This completes the proof as 
 the \textsc{$\mathcal{B}_2$-Id Code} problem is
 NP-complete (due to Corollary~\ref{cor 2-regular bipartite}), 
  $C^*$ is an oriented identifying code of $G^*$ if and only if $C$ is an oriented identifying code of $G$, 
  and $|V(G^*)| = O(|V(G)|)$.
\end{proof}

\noindent \textit{Proof of Theorem~\ref{thm main}.} Follows directly from Theorems~\ref{thm poly}, \ref{thm 2-regular NP-c}, and~\ref{thm d-regular NP-c}.

\bibliographystyle{abbrv}

\bibliography{biblio}
\end{document}